\documentclass[journal]{IEEEtran}
\usepackage{amssymb}
\usepackage{amsmath,amssymb,amsfonts}
\usepackage{multirow}
\usepackage{graphicx}
\usepackage{textcomp}
\usepackage{amsthm}
\renewcommand{\vec}[1]{\boldsymbol{#1}}
\usepackage{setspace}

\usepackage{graphicx}
\usepackage{algorithm}
\usepackage{algorithmic}
\usepackage{amsmath}
\usepackage{color}
\usepackage{amsfonts}

\newtheorem{thm}{Theorem}
\newtheorem{lem}{Lemma}

\newtheorem{defn}{Definition}

\newtheorem{rem}{Remark}
\newtheorem{pro}{Problem}

\ifCLASSINFOpdf
\else
\fi
\begin{document}
\title{Continuous Activity Maximization in Online Social Networks}


\author{Jianxiong Guo,
	Tiantian Chen,
	Weili Wu,~\IEEEmembership{Member,~IEEE}
	\thanks{J. Guo, T. Chen and W. Wu are with the Department
		of Computer Science, Erik Jonsson School of Engineering and Computer Science, Univerity of Texas at Dallas, Richardson, TX, 75080 USA
		
		E-mail: jianxiong.guo@utdallas.edu}% <-this 
	\thanks{Manuscript received April 19, 2005; revised August 26, 2015.}}


\markboth{Journal of \LaTeX\ Class Files,~Vol.~14, No.~8, August~2015}%
{Shell \MakeLowercase{\textit{et al.}}: Bare Demo of IEEEtran.cls for IEEE Journals}



\maketitle

\begin{abstract}
Activity maximization is a task of seeking a small subset of users in a given social network that makes the expected total activity benefit maximized. This is a generalization of many real applications. In this paper, we extend activity maximization problem to that under the general marketing strategy $\vec{x}$, which is a $d$-dimensional vector from a lattice space and has probability $h_u(\vec{x})$ to activate a node $u$ as a seed. Based on that, we propose the continuous activity maximization (CAM) problem, where the domain is continuous and the seed set we select conforms to a certain probability distribution. It is a new topic to study the problem about information diffusion under the lattice constraint, thus, we address the problem systematically here. First, we analyze the hardness of CAM and how to compute the objective function of CAM accurately and effectively. We prove this objective function is monotone, but not DR-submodular and not DR-supermodular. Then, we develop a monotone and DR-submodular lower bound and upper bound of CAM, and apply sampling techniques to design three unbiased estimators for CAM, its lower bound and upper bound. Next, adapted from IMM algorithm and sandwich approximation framework, we obtain a data-dependent approximation ratio. This process can be considered as a general method to solve those maximization problem on lattice but not DR-submodular. Last, we conduct experiments on three real-world datasets to evaluate the correctness and effectiveness of our proposed algorithms.
\end{abstract}

\begin{IEEEkeywords}
Activity Maximization, Lattice, DR-submodular, Social Networks, Sampling Techniques, Sandwich Approximation Framework, Approximation Algorithm
\end{IEEEkeywords}

% \fi
\IEEEpeerreviewmaketitle

\section{Introduction}
\IEEEPARstart{T}{he} online social platforms, such as Twitter, WeChat, Facebook and LinkedIn, were developing quickly in recent years, and gradually become a mainstream way to communicate and make friends. More and more people share their what they see and discuss some hot issues at the moment in these platforms. The relationships among the users in these social platforms can be represented by social networks, and information can be spread rapidly through the edges in social networks. Based on that, Influence Maximization (IM) considers the problem: selects a subset of users for an information cascade to maximize the expected follow-up adoptions (influence spread). It is a mathematical generalization of plenty of real scenarios, such as viral marketing, rumor blocking and profit maximization. In the kempe et al.'s seminal work \cite{kempe2003maximizing}, two widely accepted diffusion models were proposed, IC-model (Independent Cascade model) and (LT-model) Linear Threshold model, where IC-model is relied on peer-to-peer communication but LT-model considers the total influence from user's neighbors. Then, they showed the IM problem is NP-hard, and its objective function is monotone and submodular under the IC/LT-model, and simple greedy algorithm can achieve $(1-1/e)$-approximation \cite{nemhauser1978analysis}. In order to solve its efficiency problem, there were lots of scalable IM algorithms proposed, heuristic algorithms \cite{chen2009efficient} \cite{chen2010scalable} \cite{goyal2011simpath} \cite{jung2012irie} and approximate algorithms that improve the Monte Carlo simulations \cite{leskovec2007cost} \cite{borgs2014maximizing} \cite{tang2014influence} \cite{tang2015influence} \cite{guo2019novel} \cite{8850214} \cite{8952599}.

Motivated by IM, more interested and real problems emerged and were studied. Wang et al. \cite{wang2017activity} considered to maximize the expected total activity strength about the target issue in online social networks and proposed activity maximization problem. The activity maximization aims to maximize the total activity strength (activity benefit) associated with those edges between influenced users given a seed set. Different from IM, maximized expected influenced users do not mean that total activity strength is maximized because different edges are associated with different activity strength. In addition, they have proved the objective function of activity maximization is NP-hard, monotone, but not submodular and not supermodular \cite{wang2017activity}, and gave us a sandwich approximation framework to get an approximate solution by approximating its upper bound and lower bound.

Later, Kempe et al. \cite{kempe2015maximizing} considered a more general case that using a marketing strategy instead of the seed set. This marketing strategy is denoted by $\vec{x}=(x_1,x_2,\cdots,x_d)$ where each strategy $j$ takes value $x_j$, and for each node $u$, it will be activated as a seed with probability $h_u(\vec{x})$. Thus, the seed set is not deterministic, but activated probabilistically according to a marketing strategy. In this paper, we consider the activity maximization problem under such general marketing strategy. We propose the continuous activity maximization (CAM), which is to find the optimal marketing strategy $\vec{x}^*$ such that the expected activity benefit can be maximized subject to the budget constraint $|\vec{x}|\leq k$. In the real world, the companies often adopt some non-deterministic marketing strategies, such as discounts, coupons, rewards and propagandas, and the promotion results on different individuals are random and distinct. Therefore, CAM is more realistic and generalized than traditional activity maximization problem.

In this paper, we consider the marketing strategy $\vec{x}$ taken from discretized lattice $\mathcal{X}$ with granularity $t$, and the hardness of CAM is discussed. We show that CAM is NP-hard under the IC/LT-model. Given a marketing strategy $\vec{x}$, computing the expected activity benefit is \#P-hard. Since it is not easy to compute the expected activity benefit with respect to a given marketing strategy $\vec{x}$, we provide an equivalent method to compute it by creating a constructed graph, and running Monte Carlo simulations on this constructed graph. Then, we show that the objective function of CAM problem is monotone, but not DR-submodular and not DR-supermodular. DR-submodularity \cite{soma2015generalization} is the diminishing return property extended from set to lattice. If a function defined on lattice is DR-submodular, a $(1-1/e)$-approximation can be obtained by the simple greedy algorithm. In order to find a valid approximate solution, we construct a lower bound and upper bound that are monotone and DR-submodular. Similarly, we show that maximizing this lower bound and upper bound is NP-hard as well and computing their exact value is \#P-hard under the IC/LT-model. For IM problem, the computational cost of greedy algorithm with Monte Carlo simulations is not acceptable, to our CAM problem, the scalability could be worse than IM because the strategy space is larger and the greedy iterative times should be $k/t$ given a budget $k$ and granularity $t$. Thus, based on reverse influence sampling (RIS) \cite{borgs2014maximizing} \cite{tang2014influence} \cite{tang2015influence}, we obtain unbiased estimators for the CAM problem and its lower bound based on RE-sampling, for its upper bound based on RN-sampling. The adaption of RIS to CAM is determined by the partial coverage of the collection of RE-sampling. From this, we design a general scalable algorithm to solve CAM problem, its upper bound and lower bound adapted from IMM algorithm \cite{tang2015influence} for IM problem. We obtain a data-dependent approximation ratio by combining them with the sandwich approximation framework finally. Summarizing our contributions as follows:
\begin{enumerate}
	\item This is the first to study activity maximization problem under the general marketing strategy (lattice constraint). In this paper, a new problem, named CAM, is proposed and its objective function is proved to be monotone, but not DR-submodular and DR-supermodular.
	\item To estimate the expected activity benefit with respect to marketing strategy $\vec{x}$, it could be done on a constructed graph by Monte Carlo simulations.
	\item We obtain a lower bound and upper bound of CAM, which are monotone and DR-submodular.
	\item We design unbiased estimators for CAM and its lower/upper bound based on RE/RN-sampling. Adapted from IMM algorithm and sandwich approximation framework, a data-dependent approximation ratio can be obtained. It is the first time to consider such problems on lattice constraint.
	\item The effectiveness and correctness of our proposed algorithms are tested and verified by several datasets of real-word social networks.
\end{enumerate}

\textbf{Organization:} Sec. \uppercase\expandafter{\romannumeral2} introduces the related work. Sec. \uppercase\expandafter{\romannumeral3} is dedicated to formulate the problem.. The properties of CAM problem and upper/lower bound are presented in Sec. \uppercase\expandafter{\romannumeral4} and Sec. \uppercase\expandafter{\romannumeral5}. Sec. \uppercase\expandafter{\romannumeral6} is the sampling techniques and algorithm design for CAM. Experiments are presented in Sec. \uppercase\expandafter{\romannumeral7} and \uppercase\expandafter{\romannumeral8} is the conclusion.

\section{Related Work}
Viral marketing was first studied systematically by Domingos Richardson \cite{domingos2001mining} \cite{richardson2002mining}, and they proposed the concept of customers' the value and used markov random fields to model the process of viral marketing. Kempe et al. \cite{kempe2003maximizing} formulated IM to a combinatorial optimization problem, proposed two discrete diffusion model, generalized them to triggering model, and gave us a greedy algorithm with the constant approximation ratio. Chen et al. followed kempe's work, and proved it is \#P-hard to compute the exact influence spread for a given seed set under the IC-model \cite{chen2010scalable} and the LT-model \cite{chen2010scal}. To tackle this problem, Monte Carlo simulations were adopted as a general method, but the running time was too slow to apply to larger real networks. Subsequently, to attempt to improve the low efficiency of Monte Carlo simulations, plenty of researchers made effort, for instance, Leskovec et al. proposed a CELF algorithm \cite{leskovec2007cost} implemented by a lazy forward evaluation, avoiding redundant computation by exploiting its submodularity. Adapted from CELF, CELF++ reduced its time complexity further. Until the emergence of RIS, it opened a new door for us. Brogs et al. \cite{borgs2014maximizing} proposed the concept of reverse influence sampling (RIS) firstly, which is scalable in practice and guaranteed theoretically at the same time. Then, a series of efficient randomized algorithms were arisen, such as TIM/TIM+ \cite{tang2014influence}, IMM \cite{tang2015influence}. They were scalable algorithms to solve the IM problem with $(1-1/e-\varepsilon)$-approximation and can be adapted to other relative problems.

DR-submodular maximization problem on lattice attracted more and more researchers' attention recently. Soma et al. \cite{soma2015generalization} generalized the diminishing return property on the integer lattice firstly and solved submodular cover problem with a bicriteria approximation algorithm. Relied on gradient methods, Hassani et al. \cite{hassani2017gradient} addressed monotone continuous DR-submodular maximization effectively, but assumed that the function is continuous and differentiable. On integer lattice, Soma et al. \cite{soma2018maximizing} studied the problem of maximizing monotone DR-submodular exhaustively, where they designed algorithms with $(1-1/e-\varepsilon)$-approximation under the cardinality, polymatroid and knapsack constraint. Simultaneously, they \cite{soma2017non} considered non-monotone DR-submodular maximization over the integer lattice, and presented a $1/(2+\varepsilon)$-approximate algorithm within polynomial time. Optimal budget allocation was a typical application of the DR-submodular maximization, and was studied systematically \cite{soma2014optimal} \cite{maehara2015budget} \cite{miyauchi2015threshold} \cite{hatano2016adaptive}. To social networks, Chen et al. \cite{chen2018scalable} investigated IM problem over the lattice, whose objective function is monotone and DR-submodular. Following that, we study the activity maximization over lattice, different from IM, our objective function is monotone but not DR-submodular, which is the main contributions of this paper.

\section{Problem Formulation}
In this section, we describe the influence model, some preliminary knowledges, and then formulate the continuous activity maximization problem.

\subsection{Influence Model and Realization}
A social network is represented by a directed graph $G=(V,E)$ where $V$, $|V|=n$, denotes the set of (nodes) users, and $E$, $|E|=m$, denotes the set of directed edges which describe the relationship between users. For each edge $(u,v)\in E$, we say $u$ (resp. $v$) is an incoming neighbor (resp. an outgoing neighbor) of $v$ (resp. $u$). For each node $v\in V$, $N^-(v)$ (resp. $N^+(v)$) denotes the set of incoming neighbors (resp. outgoing neighbors) of node $v$, and $N(v)=N^-(v)\cup N^+(v)$. We adopt the IC-model and LT-model \cite{kempe2003maximizing}, to model the influence diffusion. Given a seed set $S$, the nodes in $S$ are activated and the other nodes are inactive, then the diffusion process repeats, and terminates until these is no new node is activated.

\begin{defn}[IC-model]
	A diffusion probabiltiy $p_{uv}\in(0,1]$ associated with each edge $(u,v)\in E$. For each node $u$ activated first at time step $t-1$, it activates each of its inctive outgoing neighbor $v$ with probability $p_{uv}$ at time step $t$.
\end{defn}
\begin{defn}[LT-model]
	Each edge $(u,v)\in E$ has a weight $b_{uv}$, and each node $v\in V$ has a threshold $\theta_v$ sampled uniformly in $[0,1]$ and $\sum_{u\in N^-(v))}b_{uv}\leq 1$. For each inactive node $v$ at time step $t-1$, it can be activated at time step $t$ if satisfying $\sum_{u\in A_{t-1}\cup N^-(v)}b_{uv}\leq\theta_v$, where $A_{t-1}$ is the set of active nodes at time step $t-1$.
\end{defn}

A realization ${\rm g}=(V,E(\rm g))$ is a subgraph of $G$ with $E({\rm g})\subseteq E$. Each edge in $E({\rm g})$ is live edge, or else it is blocked edge. Under the IC-model, we can decide whether edge $(u,v)$ is live or blocked with probability $p_{uv}$. Let $\Pr[{\rm g}]$ be the probability of ${\rm g}$ sampled from $G$ based on IC-model, that is
\begin{equation}
\Pr[{\rm g}]=\prod_{e\in E({\rm g})}p_e\prod_{e\in E\backslash E({\rm g})}\left(1-p_e\right)
\end{equation}
Under the LT-model, node $v$ chooses at most one of incoming neighbors $u$ from $N^-(v)$ such that edge $(u,v)$ appears in $E({\rm g})$. Thus, for each node $u\in N^-(v)$, $(u,v)$ appears in $E(\rm g)$ with probability $b_{uv}$ exclusively, and there is no incoming edge of $v$ in $E(\rm g)$ with probability $1-\sum_{u\in N^-(v)}b_{uv}$. We define $V'({\rm g})=\{v:\nexists(u,v)\in E({\rm g})\}$ as the node set which has no incoming edge in realization ${\rm g}$. Let $\Pr[{\rm g}]$ be the probability of ${\rm g}$ sampled from $G$ based on LT-model, that is
\begin{equation}
\Pr[{\rm g}]=\prod_{e\in E({\rm g})}b_{uv}\prod_{v\in V'}\bigg(1-\sum_{u\in N^-(v)}b_{uv}\bigg)
\end{equation}
The stochasic diffusion process on $G$ can be considered as deterministic diffusion process on ${\rm g}$ sampled from $G$.

\subsection{Problem Definition}
In the activity maximization problem, there are an activity strength $A_{uv}\in\mathbb{R}_+$ associated with each edge $(u,v)\in E$. $A_{uv}$ means that the benefit or profit between user $u$ and user $v$ if they are both active \cite{wang2017activity}. Given a social graph $G=(V,E)$, an influnece model, and seed set $S$, we define $I(S)$ as the set of activated nodes after the diffusion terminates and $G[I(S)]=(I(S),E[I(S)])$ as the induced subgraph by activated node set $I(S)$, where we have $E[I(S)]=\{(u,v)\in E:u\in I(S)\land v\in I(S)\}$. Given the seed set $S$, the activity function of the activity maximization problem \cite{wang2017activity} is
\begin{equation}
f_d(S)=\mathbb{E}\left[\sum_{(u,v)\in E[I(S)]}A_{uv}\right]
\end{equation}
where $f_d(S)$ is the expected activity benefit of final active nodes for the diffusion starting from $S$. The task of activity maximization is to select at most $k$ seed nodes to maximize the expected activity benefit, i.e., to find $S^*=\arg\max_{S\subseteq V,|S|\leq k}f_d(S)$.

In this paper, we extend the activity maximization problem with general marketing strategy \cite{kempe2015maximizing}, which is a $d$-dimensional vector $\vec{x}=(x_1,x_2,...,x_d)\in\mathbb{R}_+^d$. Each component $x_i$, $i\in[d]=\{1,2,...,d\}$, corresponds to the investment to marketing action $M_i$. Given a marketing strategy $\vec{x}$, the probability that node $u\in V$ is selected as a seed is denoted by strategy function $h_u(\vec{x})$, where $h_u(\vec{x})\in[0,1]$. Thus, different from previous definition, the seed set under the general marketing strategy is stochastic, not deterministic. Given a marketing strategy $\vec{x}$, the probability we select $S\subseteq V$ according to $\vec{x}$ as the seed set is
\begin{equation}
\Pr[S|\vec{x}]=\prod_{u\in S}h_u(\vec{x})\cdot\prod_{v\in V\backslash S}(1-h_v(\vec{x}))
\end{equation}
where $\Pr[S|\vec{x}]$ is the probability that exactly nodes in $S$ are selected as seeds but not in $S$ are not selected as seeds under the marketing strategy $\vec{x}$, which is because each node is select as a seed indetpendently. Then, the activity function now is
\begin{flalign}
	f_c(\vec{x})&=\sum_{S\subseteq V}\Pr[S|\vec{x}]\cdot f_d(S)\\
	&=\sum_{S\subseteq V}f_d(S)\cdot\prod_{u\in S}h_u(\vec{x})\cdot\prod_{v\in V\backslash S}(1-h_v(\vec{x}))
\end{flalign}
\begin{rem}
	We can address marketing strategy $\vec{x}$ in a discretized manner with granularity parameter $t$, where each component $x_i$ takes discretized value $\{0,t,2t,\cdots\}$. These set of vectors is called as lattice $\mathcal{X}$, where $\mathcal{X}=\{0,t,2t,\cdots\}^d$.
\end{rem}
Now, we define the continuous activity maximization (CAM) problem as follows:
\begin{pro}[Continuous Activity Maximization]
	Given a social network $G=(V,E)$ with a influence model, a budget $k$, a marketing strategy functions $h_u(\cdot)$ for each user $u$, CAM aims to find an optimal marketing strategy $\vec{x}$ such that the expected activity benefit can be maximized, that is
	\begin{equation}
	\vec{x}^*=\arg\max_{\vec{x}\in\mathcal{X},|\vec{x}|\leq k}f_c(\vec{x})
	\end{equation}
	where consider the marketing strategy $\vec{x}$ under the budget constraint: $|\vec{x}|=\sum_{i\in[d]}x_i\leq k$. Here, each configuration satisfying $\vec{x}\in\mathcal{X}$ and $|\vec{x}|\leq k$ is called as a feasible solution.
\end{pro}
\noindent
To make the context clear, we refer to the problem that finding $S^*=\arg\max_{S\subseteq V,|S|\leq k}f_d(S)$ as discrete activity maximization (DAM).

\section{Properties of CAM}
In this section, we discuss the hardness, submodularity and approximability of our CAM problem.
\subsection{Hardness}
In order to show the hardness, we can start from a classical NP-hard problem, Set Cover problem, and reduce MC to our CAM problem in polynomial time.

\begin{thm}
	The CAM problem is NP-hard under the IC-model and the LT-model.
\end{thm}
\begin{proof}
	We assume that $\mathcal{X}=\{0,1\}^n$ and $h_v(\vec{x})=x_v$, that is, $v$ is selected as a seed if and only if $x_v=1$. Now, marketing strategy $\vec{x}$ is the characteristic vector of the seed set, and CAM problem can be reduced to DAM problem trivally. It has been proved in \cite{wang2017activity} that DAM is NP-hard under the IC-model and LT-model by reducing from the set cover problem. Thus, CAM is more general, and it is NP-hard by inheriting the NP-hardness of DAM.
\end{proof}
It is known that under the IC-model and LT-model, computing influence spread is \#P-hard \cite{chen2010scalable} \cite{chen2010scal}. Given a marketing strategy $\vec{x}$, the hardness of computing $f_c(\vec{x})$, that is
\begin{lem}
	Given a marketing strategy $\vec{x}$, computing $f_c(\vec{x})$ by Equation (5) is \#P-hard.
\end{lem}
\begin{proof}
	Similar to the proof of Theorem 1, CAM can be reduced to DAM problem by setting $\mathcal{X}=\{0,1\}^n$ and $h_v(\vec{x})=x_v$. Based on Equation (3), computing $f_d(S)$ is equivalent to compute $\mathbb{E}[I(S)]$, thus, computing $f_d(S)$ is \#P-hard. Except for this special case, the computation of $f_c(S)$ is harder than $f_d(S)$, we hare computing $f_c(S)$ is \#P-hard.
\end{proof}
Monte Carlo simulation can be used to estimate $f_c(\vec{x})$ because it is the expectation of $f_d(\vec{x})$ over the random variable $S$. We need to sample $S$ according to distribution $\vec{x}$.
\begin{lem}
	Provided that we have value oracle that returns the activity benefit $f_d(S)$ given a seed set $S$, we can obtain a $(\gamma,\delta)$-Estimation of $f_c(\vec{x})$ by sampling $S$ according to $\vec{x}$ at least $\frac{\alpha^2\ln(2/\delta)}{2\gamma^2\beta^2}$ times, where $\alpha=\sum_{(u,v)\in E}[A_{uv}]$ and $\beta=\sum_{e\in E}[h_u(\vec{x})h_v(\vec{x})\cdot A_{uv}]$.
\end{lem}
\begin{proof}
	According to Equation (5), we can estimate $f_c(\vec{x})$ with the help of Monte Carlo simultions, denoted by $\dot{f}_c(\vec{x})$ and based on Hoeffding's inequality, we have
	\begin{equation*}
		\Pr\left[\left|\dot{f}_c(\vec{x})-f_c(\vec{x})\right|\geq\gamma f_c(\vec{x})\right]\leq 2e^{-\frac{2r\gamma^2(f_c(\vec{x}))^2}{\alpha^2}}
	\end{equation*}
	where $r$ is the number of Monte Carlo simultions and $f_d(S)\in[0,\alpha]$. Then, we consider the lower bound of $f_c(\vec{x})$. For each edge $(u,v)\in E$, the probability of both $u$ and $v$ are active is at least $h_u(\vec{x})h_v(\vec{x})$, thus, we have $f_c(\vec{x})\geq\sum_{e\in E}[h_u(\vec{x})h_v(\vec{x})\cdot A_{uv}]$. Therefore, we can set $r\geq\frac{\alpha^2\ln(2/\delta)}{2\gamma^2\beta^2}$ that establishing $\Pr[|\dot{f}_c(\vec{x})-f_c(\vec{x})|\geq\gamma f_c(\vec{x})]\leq\delta$.
\end{proof}

Unfortunately, it is not easy to compute the activity benefit $f_d(S)$ given a seed set $S$. Thus, we need to address this problem by other techniques. First, we establish an equivalent relationship bewteen $f_d(\cdot)$ and $f_c(\cdot)$. Given a social graph $G=(V,E)$ and a marketing strategy $\vec{x}$, we create a constructed graph $\widetilde{G}=(\widetilde{V},\widetilde{E})$ by adding a new node $\widetilde{u}$ and a new directed edge $(\widetilde{u},u)$ for each node $u$ in $V$ to $G$, where $(\widetilde{u},u)$ is with activation probability $p_{\widetilde{u}u}=h_u(\vec{x})$ in IC-model and weight $b_{\widetilde{u}u}=h_u(\vec{x})$ in LT-model. Then, we can observe that
\begin{equation}
f_c(\vec{x}|G)=f_d(\widetilde{V}-V|\widetilde{G})-\sum_{u\in V}[h_u(\vec{x})\cdot A_{\widetilde{u}u}]
\end{equation}
where $f_c(\vec{x}|G)$ means computing $f_c(\vec{x})$ under the graph $G$. We set the activity strength $A_{\widetilde{u}u}=0$ for each node $u$ in $V$, then we have $f_c(\vec{x}|G)=f_d(\widetilde{V}-V|\widetilde{G})$. Now, we can compute $f_d(\widetilde{V}-V|\widetilde{G})$ instead of $f_c(\vec{x}|G)$ when we are required to get the value of $f_c(\vec{x}|G)$.

\begin{thm}
	Given a social graph $G=(V,E)$ and a marketing strategy $\vec{x}$, the total running time to get a $(\gamma,\delta)$-Estimation of $f_c(\vec{x})$ is $O\left(\frac{(m+n)\alpha^2\ln(2/\delta)}{2\varepsilon^2\beta^2}\right)$, where $\alpha=\sum_{(u,v)\in E}[A_{uv}]$ and $\beta=\sum_{e\in E}[h_u(\vec{x})h_v(\vec{x})\cdot A_{uv}]$.
\end{thm}
\begin{proof}
	From the Equation (8), we have $f_c(\vec{x}|G)=f_d(\widetilde{V}-V|\widetilde{G})$. According to Equation (3), we can estimate $f_d(\widetilde{V}-V|\widetilde{G})$ by Monte Carlo simulations. Denoted by $S'=\widetilde{V}-V$, and based on Hoeffding's inequality, we have
	\begin{equation*}
		\Pr\left[\left|\dot{f}_d(S')-f_d(S')\right|\geq\gamma f_d(S')\right]\leq 2e^{-\frac{2r\gamma^2(f_d(S'))^2}{\alpha^2}}
	\end{equation*}
	where $r$ is the number of Monte Carlo simultions and $\sum_{(u,v)\in\widetilde{E}[I(S')]}A_{uv}\in[0,\alpha]$. Then, we consider the lower bound of $f_d(S')$. Similar to Lemma 2, we have $f_d(S')\geq\sum_{e\in E}[h_u(\vec{x})h_v(\vec{x})\cdot A_{uv}]$ as well. To achieve a $(\gamma,\delta)$-Estimation of $f_d(S')$, the number of Monte Carlo simultions is at least $\frac{\alpha^2\ln(2/\delta)}{2\gamma^2\beta^2}$. Each Monte Carlo simulation takes $O(m+n)$ running time in constructed graph $\widetilde{G}$. Thus, we have a $(\gamma,\delta)$-Estimation of $f_c(\vec{x}|G)$ in $O\left(\frac{(m+n)\alpha^2\ln(2/\delta)}{2\gamma^2\beta^2}\right)$ running time.
\end{proof}
\begin{rem}
	From the Lemma 2 and Theorem 2, we can know that computing $f_c(\vec{x})$ on $G$ is equivalent to compute $f_d(\widetilde{V}-V)$ on constructed graoh $\widetilde{G}$, which give us an efficient technique to estimate the value $f_c(\vec{x})$ by use of Monte Carlo simulations.
\end{rem}

\subsection{Modularity of Objective Functions}
In order to address CAM problem, a intuitive method is to use the greedy algorithm that can obtain a constant approximation ratio depended on the diminishing return property. We say that A set function $f:2^V\rightarrow\mathbb{R}$ is monotone if $f(S)\leq f(T)$ for all $S\subseteq T\subseteq V$, and submodular if $f(S\cup\{u\})-f(S)\geq f(T\cup\{u\})-f(T)$ for all $S\subseteq T \subseteq V$ and $u\in V\backslash T$. Conversely, if $f(S\cup\{u\})-f(S)\leq f(T\cup\{u\})-f(T)$ for all $S\subseteq T \subseteq V$ and $u\in V\backslash T$, we say $f$ is supermodular. Soma et al. \cite{soma2015generalization} extended the submodularity and the diminishing return property to functions defined on the lattice, that is referred to as the DR-submodular property. To our CAM problem, for two vectors $x,y\in\mathcal{X}$, a function $g:\mathcal{X}\rightarrow\mathbb{R}$ is monotone if $g(\vec{x})\leq g(\vec{y})$ for all $\vec{x}\leq\vec{y}$, and DR-submodular if $g(\vec{x}+t\vec{e}_i)-g(\vec{x})\geq g(\vec{y}+t\vec{e}_i)-g(\vec{y})$ for all $\vec{x}\leq\vec{y}$ and $i\in[d]$. Conversely, if $g(\vec{x}+t\vec{e}_i)-g(\vec{x})\leq g(\vec{y}+t\vec{e}_i)-g(\vec{y})$ for all $\vec{x}\leq\vec{y}$ and $i\in[d]$, we say $g$ is DR-supermodular. Unfortunately, the objective function of CAM problem is not DR-submodular and DR-supermodular.
\begin{rem}
	Here, we assume that the strategy functions $h_u(\vec{x})$ for each $u\in V$ are monotone and DR-submodular. It is because the probability that a user agrees to be a seed increases with more investment and this marginal gain is non-increasing.
\end{rem}
\begin{thm}
	$f_c(\cdot)$ is monotone but not DR-submodular under the IC-model and the LT-model.
\end{thm}
\begin{proof}
	We prove by a counterexample, consider graph $G=(V,E)$, $V=\{v_1,v_2,v_3,v_4\}$ and $E=\{(v_1,v_2),(v_2,v_3),(v_4,v_3)\}$. By setting $\mathcal{X}=\{0,1\}^4$ and $h_v(\vec{x})=x_v$, we have $h_v(\vec{x})$ is monotone and DR-submodular. The activation probabilities in IC-model and weights in LT-model of $\{(v_1,v_2),(v_4,v_3)\}$ are set to be $1$, but $\{(v_2,v_3)\}$ is $0$. The activity strengths are all set to be $1$. Let $\vec{x}=(0,0,0,0)$ and $\vec{y}=(0,0,0,1)$, we have $f_c(\vec{x})=0$, $f_c(\vec{x}+e_1)=1$, $f_c(\vec{y})=1$ and $f_c(\vec{y}+e_1)=3$. That is $f_c(\vec{x}+e_1)-f_c(\vec{x})<f_c(\vec{y}+e_1)-f_c(\vec{y})$ where $\vec{x}\leq\vec{y}$. Therefore, $f_c(\cdot)$ is not DR-submodular.
\end{proof}

In \cite{wang2017activity}, they explained the reason why $f_d(\cdot)$ is not submodular as the "combination effect" between the new activated node with existing activated node. It can be extended to $f_c(\cdot)$ naturally.

\begin{thm}
	$f_c(\cdot)$ is monotone but not DR-supermodular under the IC-model and the LT-model.
\end{thm}
\begin{proof}
	We prove by a counterexample, consider graph $G=(V,E)$, $V=\{v_1,v_2,v_3,v_4\}$ and $E=\{(v_2,v_1),(v_2,v_3),(v_3,v_4)\}$. By setting $\mathcal{X}=\{0,1\}^4$ and $h_v(\vec{x})=x_v$, we have $h_v(\vec{x})$ is monotone and DR-submodular. The activation probabilities in IC-model, weights in LT-model and activity strengths are all set to be $1$. Let $\vec{x}=(0,0,0,0)$ and $\vec{y}=(0,0,1,0)$, we have $f_c(\vec{x})=0$, $f_c(\vec{x}+e_2)=3$, $f_c(\vec{y})=1$ and $f_c(\vec{y}+e_2)=3$. That is $f_c(\vec{x}+e_2)-f_c(\vec{x})>f_c(\vec{y}+e_2)-f_c(\vec{y})$ where $\vec{x}\leq\vec{y}$. Therefore, $f_c(\cdot)$ is not DR-supermodular.
\end{proof}

\section{Upper and Lower Bound}
In this section, we design an upper bound and a lower bound for our objective function $f_c(\cdot)$, and discuss the properties of them.
\subsection{Bounds Definition}
According to the activity function of CAM problem, Eqaution (5), in order to get an upper bound and a lower bound of $f_c(\cdot)$, we firstly need to get both bounds of DAM problem $f_d(\cdot)$. Wang et al. \cite{wang2017activity} pointed out that the non-submodularity of $f_d(\cdot)$ is derived from the "combination effect". Thus, for a lower bound, only those edges whose two endpoints are influenced by the cascade from the same seed node. we denote by $\underline{f_d}$ the lower bound of $f_d$, that is
\begin{equation}
\underline{f_d}(S)=\mathbb{E}\left[\sum_{(u,v)\in\bigcup_{x\in S}E[I(x)]}A_{uv}\right]
\end{equation}
where $E[I(x)]$ is the edges of induced subgraph by activated node set $I(x)$. Given a seed set $S$, we have $\underline{f_d}(S)\leq f_d(S)$ because it neglects those edges whose endpoints can not be activated by the different seed nodes. Then, we denote by $\overline{f_d}$ the upper bound of $f_d$, that is
\begin{equation}
\overline{f_d}(S)=\mathbb{E}\left[\sum_{u\in V[I(S)]}\sum_{v\in N(u)}\frac{A_{uv}}{2}\right]
\end{equation}
where $V[I(S)]$ is the nodes of induced subgraph by activated node set $I(S)$. Given a seed set $S$, we have $\overline{f_d}(S)\geq f_d(S)$ because we considers each active node contributes to half of activity strength associated to those edges connected to it. Thus, for each edge, it is not mandatory to require both of its endpoints are activated.

According to the above bounds of $f_d$, we can obtain the upper bound and lower bound of the activity function of CAM problem by the same way. From Equation (6), we denote by $\underline{f_c}$ the lower bound of $f_c$, that is
\begin{equation}
\underline{f_c}(\vec{x})=\sum_{S\subseteq V}\underline{f_d}(S)\cdot\prod_{u\in S}h_u(\vec{x})\cdot\prod_{v\in V\backslash S}(1-h_v(\vec{x}))
\end{equation}
denote by $\overline{f_c}$ the upper bound of $f_c$, that is
\begin{equation}
\overline{f_c}(\vec{x})=\sum_{S\subseteq V}\overline{f_d}(S)\cdot\prod_{u\in S}h_u(\vec{x})\cdot\prod_{v\in V\backslash S}(1-h_v(\vec{x}))
\end{equation}
Given a marketing strategy $\vec{x}$, we have $\underline{f_c}(\vec{x})\leq f_c(\vec{x}) \leq\overline{f_c}(\vec{x})$ because $\underline{f_c}(\vec{x})$ (resp. $\overline{f_c}(\vec{x})$) is the linear combination of $\underline{f_d}(S)$ (resp. $\overline{f_d}(S)$). Thus, we can conclude that $\underline{f_d}(S)\leq f_d(S) \leq\overline{f_d}(S)$ means $\underline{f_c}(S)\leq f_c(S) \leq\overline{f_c}(S)$.

\subsection{Properties of the Bounds}
Lu et al. \cite{lu2015competition} provided us with a idea where we can obtain an approximate solution of CAM problem by maximizing its the upper bound and lower bound. As we know, by setting $\mathcal{X}=\{0,1\}^n$ and $h_v(\vec{x})=x_v$, the CAM problem can be reduced to DAM problem. Similarly, maximizing the $\underline{f_c}(\vec{x})$ (resp. $\overline{f_c}(\vec{x})$) can also be reduced maximixing the $\underline{f_d}(S)$ (resp. $\overline{f_d}(S)$) under this special case, which inherits its NP-hardness. Because of maximizing the $\underline{f_d}(\cdot)$ and $\overline{f_d}(\cdot)$ is NP-hard \cite{wang2017activity}, it is natural to have
\begin{thm}
	Maximizing the lower bound $\underline{f_c}(\cdot)$ is NP-hard under the IC-model and the LT-model.
\end{thm}
\begin{thm}
	Maximizing the upper bound $\overline{f_c}(\cdot)$ is NP-hard under the IC-model and the LT-model.
\end{thm}
Even though that, the lower bound $\underline{f_d}(\cdot)$ and upper bound $\overline{f_d}(\cdot)$ of DAM is submodular.
\begin{lem}[\cite{wang2017activity}]
	The lower bound $\underline{f_d}(\cdot)$ is monotone and submodular, but computing it given a seed set $S$ is \#P-hard under the IC-model and the LT-model.
\end{lem}
\begin{lem}[\cite{wang2017activity}]
	The upper bound $\overline{f_d}(\cdot)$ is monotone and submodular, but computing it given a seed set $S$ is \#P-hard under the IC-model and the LT-model.
\end{lem}
\noindent
Then, the submodularity of $\underline{f_d}(\cdot)$ $($resp, $\overline{f_d}(\cdot)$$)$ can be correlated to the DR-submodularity of $\underline{f_c}(\cdot)$ $($resp, $\overline{f_c}(\cdot)$$)$. Let us look at the following Lemma:
\begin{lem}
	Given a set function $f:2^V\rightarrow\mathbb{R}$ and a function $g:\mathcal{X}\rightarrow\mathbb{R}$, they satisfies that
	\begin{equation}
	g(\vec{x})=\sum_{S\subseteq V}f(S)\cdot\prod_{u\in S}h_u(\vec{x})\cdot\prod_{v\in V\backslash S}(1-h_v(\vec{x}))
	\end{equation}
	When $h_u(\vec{x})$ for each $u\in V$ are monotone and DR-submodular, if $f(\cdot)$ is monotone and submodular, then $g(\cdot)$ is monotone and DR-submodular.
\end{lem}
\begin{proof}
	This lemma is an indirect corollary from the section 7 of \cite{kempe2015maximizing}, but there is a typo over there, and we fix and rearrange here. We denote $\alpha(u)=h_u(\vec{x}+t\vec{e}_j)-h_u(\vec{x})$ and $\beta(u,S)=\prod_{i<u,i\in S}h_i(\vec{x}+t\vec{e}_j)\cdot\prod_{i<u,i\notin S}(1-h_i(\vec{x}+t\vec{e}_j))\cdot\prod_{i<u,i\in S}h_i(\vec{x})\cdot\prod_{i<u,i\notin S}(1-h_i(\vec{x}))$. Thus, we have $g(\vec{x}+t\vec{e}_i)-g(\vec{x})=\sum_{S\subseteq V}f(S)\cdot(\prod_{u\in S}h_u(\vec{x}+t\vec{e}_j)\cdot\prod_{u\in V\backslash S}(1-h_u(\vec{x}+t\vec{e}_j))-\prod_{u\in S}h_u(\vec{x})\cdot\prod_{u\in V\backslash S}(1-h_u(\vec{x})))=\sum_{S\subseteq V}f(S)\cdot(\sum_{u\in S}\alpha(u)\cdot\beta(u,S)-\sum_{u\in V\backslash S}\alpha(u)\cdot\beta(u,S))=\sum_{u\in V}(\alpha(u)\cdot\sum_{S:u\in V\backslash S}(f(S\cup\{u\})-f(S))\cdot\beta(u,S))$. Then, we study the difference $(g(\vec{x}+t\vec{e}_i)-g(\vec{x}))-(g(\vec{y}+t\vec{e}_i)-g(\vec{y}))$ for $\vec{x}\leq\vec{y}$, and show it is non-negative, whose techniques are similar to the section 7 of \cite{kempe2015maximizing}.
\end{proof}
Based on Lemma 3, Lemma 4 and Lemma 5, the following theorems can be introduced directly, that is
\begin{thm}
	The lower bound $\underline{f_c}(\cdot)$ is monotone and DR-submdoualr, but computing it given a marketing strategy $\vec{x}$ is \#P-hard under the IC-model and the LT-model.
\end{thm}
\begin{thm}
	The upper bound $\overline{f_c}(\cdot)$ is monotone and DR-submdoualr, but computing it given a marketing strategy $\vec{x}$ is \#P-hard under the IC-model and the LT-model.
\end{thm}
\noindent
Given a marketing strategy $\vec{x}$, how can we compute the value of $\underline{f_c}(\vec{x})$ and $\overline{f_c}(\vec{x})$ effectively. The same as before, Equation (8), we create a constructed graph $\widetilde{G}=(\widetilde{V},\widetilde{E})$. According to Remark 2, computing $\underline{f_c}(\vec{x})$ $($resp, $\overline{f_c}(\vec{x})$$)$ is equivalent to compute $\underline{f_d}(\widetilde{V}-V|\widetilde{G})$ $($resp, $\overline{f_d}(\widetilde{V}-V|\widetilde{G})$$)$. They can be done by user of Monte Carlo simulations.

\section{Algorithms}
Given a function $g$ on lattice $\mathcal{X}=\{0,t,2t,\cdots\}^d$ and a budget $k$, the lattice-Greedy algorithm is shown in Algorithm \ref{a1}. If this function $g$ is monotone and DR-submodular, Algorithm \ref{a1} returns a solution that achieves a $(1-1/e)$-approximation \cite{nemhauser1978analysis}. The idea of lattice-Greedy algorithm is to find the component that has the largest marginal gain, and then allocate one unit $t$ (lattice granularity) to this coordinate until the budget is exhausted. In our CAM problem, it is \#P-hard to compute the lower bound $\underline{f_c}(\vec{x})$ and the upper bound $\overline{f_c}(\vec{x})$ in IC-model and LT-model. Thus, Algorithm \ref{a1} can give us a  $(1-1/e-\varepsilon)$-approximate solution by use of Morto Carlo simulations. However, the efficiency of Monte Carlo simulations is very low, so it is not scalable. In this section, we propose the sampling technique for these objective functions such that our CAM problem is scalable based on reverse influence sampling (RIS) \cite{borgs2014maximizing}. Then, we adapt Influence Maximization with Martingale (IMM) \cite{tang2015influence} algorithm and combine it with sandwich approximation framework to solve our lattice-based problem.

\begin{algorithm}[!t]
	\caption{\textbf{lattice-Greedy $(g,\mathcal{X},k)$}}\label{a1}
	\begin{algorithmic}[1]
		\STATE Initialize: $\vec{x}=0$ and $c=0$
		\WHILE {$c<k$}
		\STATE $i^*\leftarrow\arg\max_{i\in [d]}(g(\vec{x}+t\vec{e}_i)-g(\vec{x}))$
		\STATE $\vec{x}\leftarrow\vec{x}+t\vec{e}_{i^*}$
		\STATE $c\leftarrow c+t$
		\ENDWHILE
		\RETURN $\vec{x}$
	\end{algorithmic}
\end{algorithm}

\subsection{Sampling techniques}
Given a social network $G=(V,E)$, an diffusion model (IC/LT-model), and a seed set $S$, let ${\rm g}=(V,E_{\rm g})$ be a realization sampled from a distribution, Equation (1) or Equation (2), denoted by ${\rm g}\sim G$. We denote by $R_{\rm g}(S)$ the set of nodes that are reachable from at least one node in $S$ through $E_{\rm g}$ and $R_{{\rm g}^T}(v)$ the reverse reachable set (RR-Set) \cite{tang2014influence} for node $v$ in ${\rm g}$, which is a set composed of all nodes that can reach $v$ through $E_{\rm g}$. Let $(u,v)$ be a edge sampled from the probability distribution $A_{uv}/T$ where $T=\sum_{(u,v)\in E}A_{uv}$, denoted by $(u,v)\sim E$. Then, a random edge sampling (RE-sampling) $\mu$ can be defined as follows:
\begin{enumerate}
	\item Initialize $\mu=(\emptyset,\emptyset)$
	\item Select an edge $(u,v)\in E$ with probaility $A_{uv}/T$
	\item Generate a realization ${\rm g}$ from $G$ according to the IC/LT-model
	\item Let $N_1=R_{{\rm g}^T}(u)$ and $N_2=R_{{\rm g}^T}(u)$
	\item Let $\mu=(N_1,N_2)$
	\item Return $\mu$
\end{enumerate}

Given a marketing strategy $\vec{x}$, to estimate $f_c(\vec{x})$, we have the following results, that is,
\begin{thm}
	Given $G=(V,E)$ and a marketing strategy $\vec{x}\in\mathcal{X}$, we have
	\begin{equation}
	f_c(\vec{x})=T\cdot\mathbb{E}_{\mu=(N_1,N_2)}\left[\mathcal{H}(N_1)\cdot\mathcal{H}(N_2)\right]
	\end{equation}
	where $\mu$ is a RE sampling, $T=\sum_{(u,v)\in E}A_{uv}$ and $\mathcal{H}(N_1)=1-\prod_{s\in N_1}(1-h_s(\vec{x}))$.
\end{thm}
\begin{proof}
	Given a marketing strategy $\vec{x}\in\mathcal{X}$, according to Equation (5), we can write $f_c(\vec{x})$ as $f_c(\vec{x})=$
	\begin{flalign}
		&=\mathbb{E}_{S\sim\vec{x}}[f_d(S)]\nonumber\\
		&=T\cdot\mathbb{E}_{S\sim\vec{x},\mu=(N_1,N_2)}[\mathbb{I}(S\cap N_1\neq\emptyset\land S\cap N_2\neq\emptyset)]\nonumber\\
		&=T\cdot\mathbb{E}_{\mu=(N_1,N_2)}\left[\Pr_{S\sim\vec{x}}[S\cap N_1\neq\emptyset\land S\cap N_2\neq\emptyset)]\right]\nonumber
	\end{flalign}
	Here, the domain $N_1\cup N_2$ can be considered as $(N_1\cap N_2)\cup(N_1\backslash N_2)\cup(N_2\backslash N_1)$. Thus, we have $f_c(\vec{x})=$
	\begin{flalign}
		&=T\cdot\mathbb{E}_{S\sim\vec{x},\mu=(N_1,N_2)}\left[{\Pr_{S\sim\vec{x}}[S\cap(N_1\cap N_2)\neq\emptyset]}\right.\nonumber\\
		&\left.{+\Pr_{S\sim\vec{x}}[S\cap(N_1\cap N_2)=\emptyset]\cdot\Pr_{S\sim\vec{x}}[S\cap(N_1\backslash N_2)\neq\emptyset]}\right.\nonumber\\
		&\left.{\cdot\Pr_{S\sim\vec{x}}[S\cap(N_1\backslash N_2)\neq\emptyset]}\right]\nonumber\\
		&=T\cdot\mathbb{E}_{S\sim\vec{x},\mu=(N_1,N_2)}\left[{\mathcal{H}(N_1\cap N_2)}\right.\nonumber\\
		&\left.{+(1-\mathcal{H}(N_1\cap N_2))\cdot\mathcal{H}(N_1\backslash N_2)\cdot\mathcal{H}(N_2\backslash N_1)}\right]\nonumber
	\end{flalign}
	where $\mathbb{I}(\cdot)$ is the indicator function which is equal to $1$ if $(\cdot)$ is true. Then, $\Pr_{S\sim\vec{x}}[S\cap N_1\neq\emptyset]$ is the probabilty there is at least one node in $N_1$ activated as a seed, thus, we have  $\Pr_{S\sim\vec{x}}[S\cap N_1\neq\emptyset]=1-\prod_{s\in N_1}(1-h_s(\vec{x}))=\mathcal{H}(N_1)$.
\end{proof}
\noindent
Let $M=\{\mu_1,\mu_2,\cdots,\mu_\theta\}$ be a collection of $\theta$ independent RE-sampling, by Equation (14), we have
\begin{flalign}
	\hat{f}_c(\vec{x})&=\frac{T}{\theta}\sum_{\mu=(N_1,N_2),\mu \in M}\left({\mathcal{H}(N_1\cap N_2)}\right.\nonumber\\
	&\left.{+(1-\mathcal{H}(N_1\cap N_2))\cdot\mathcal{H}(N_1\backslash N_2)\cdot\mathcal{H}(N_2\backslash N_1)}\right)
\end{flalign}
According to Theorem 9, $\hat{f}_c(\vec{x})$ is an unbiased estimator of $f_c(\vec{x})$ for any fixed $\theta$ and it is not DR-submodular as well. Similarly, for the lower bound $\underline{f_c}(\vec{x})$, we have the following results, that is,
\begin{thm}
	Given $G=(V,E)$ and a marketing strategy $\vec{x}\in\mathcal{X}$, we have
	\begin{equation}
	\underline{f_c}(\vec{x})=T\cdot\mathbb{E}_{\mu=(N_1,N_2)}\left[\mathcal{H}(N_1\cap N_2)\right]
	\end{equation}
	where $\mu$ is a RE sampling, $T=\sum_{(u,v)\in E}A_{uv}$ and $\mathcal{H}(N_1\cap N_2)=1-\prod_{s\in N_1\cap N_2}(1-h_s(\vec{x}))$.
\end{thm}
\begin{proof}
	Given a marketing strategy $\vec{x}\in\mathcal{X}$, according to Equation (11), we can write $\underline{f_c}(\vec{x})$ as $\underline{f_c}(\vec{x})=$
	\begin{flalign}
		&=\mathbb{E}_{S\sim\vec{x}}[\underline{f_d}(S)]\nonumber\\
		&=T\cdot\mathbb{E}_{S\sim\vec{x},\mu=(N_1,N_2)}[\mathbb{I}(S\cap(N_1\cap N_2)\neq\emptyset)]\nonumber\\
		&=T\cdot\mathbb{E}_{\mu=(N_1,N_2)}\left[\Pr_{S\sim\vec{x}}[(S\cap(N_1\cap N_2)\neq\emptyset)]\right]\nonumber\\
		&=T\cdot\mathbb{E}_{\mu=(N_1,N_2)}\left[\mathcal{H}(N_1\cap N_2)\right]\nonumber
	\end{flalign}
	where $\mathbb{I}(\cdot)$ is the indicator function which is equal to $1$ if $(\cdot)$ is true. Then, $\Pr_{S\sim\vec{x}}[S\cap(N_1\cap N_2)\neq\emptyset]$ is the probabilty there is at least one node in $N_1\cap N_2$ activated as a seed because it requires that the endpoints of an edge can be activated by the same seed node, thus, we have $\Pr_{S\sim\vec{x}}[S\cap(N_1\cap N_2)\neq\emptyset]=1-\prod_{s\in(N_1\cap N_2)}(1-h_s(\vec{x}))=\mathcal{H}(N_1\cap N_2)$.
\end{proof}
\noindent
By Equation (16), we have
\begin{equation}
\underline{\hat{f}_c}(\vec{x})=\frac{T}{\theta}\sum_{\mu=(N_1,N_2),\mu \in M}\left(\mathcal{H}(N_1\cap N_2)\right)
\end{equation}

For the upper bound $\overline{f_c}(\vec{x})$, the sampling technique is a litte different. Shown as Equation (10), the upper bound is a weighted influence maxization on lattice. Let $u$ be a node sampled from the probability distribution $w(u)/W$ where $w(u)=\sum_{v\in N(u)}A_{uv}/2$ and $W=\sum_{u\in V}w(u)$, denoted by $u\sim V$. Then, a random node sampling (RN-sampling) $\nu$ can be defined as follows:
\begin{enumerate}
	\item Initialize $\nu=(\emptyset,\emptyset)$
	\item Select an node $u\in V$ with probaility $w(u)/W$
	\item Generate a realization ${\rm g}$ from $G$ according to the IC/LT-model
	\item Let $\nu=R_{{\rm g}^T}(u)$
	\item Return $\nu$
\end{enumerate}

Given a marketing strategy $\vec{x}$, to estimate $\overline{f_c}(\vec{x})$, we have the following results, that is,
\begin{thm}
	Given $G=(V,E)$ and a marketing strategy $\vec{x}\in\mathcal{X}$, we have
	\begin{equation}
	\overline{f_c}(\vec{x})=W\cdot\mathbb{E}_{\nu}\left[\mathcal{H}(\nu)\right]
	\end{equation}
	where $\nu$ is a RN sampling, $W=\sum_{u\in V}w(u)$ and $\mathcal{H}(\nu)=1-\prod_{s\in\nu}(1-h_s(\vec{x}))$.
\end{thm}
\begin{proof}
	Given a marketing strategy $\vec{x}\in\mathcal{X}$, according to Equation (18), we can write $\overline{f_c}(\vec{x})$ as $\overline{f_c}(\vec{x})=\mathbb{E}_{S\sim\vec{x}}[\overline{f_d}(S)]=W\cdot\mathbb{E}_{S\sim\vec{x},\nu}[\mathbb{I}(S\cap\nu\neq\emptyset)]=W\cdot\mathbb{E}_{\nu}[\Pr_{S\sim\vec{x}}[(S\cap\nu\neq\emptyset)]]=W\cdot\mathbb{E}_{\nu}\left[\mathcal{H}(\nu)\right]$. Then, $W\cdot\mathbb{E}_{S\sim\vec{x},\nu}[\mathbb{I}(S\cap\nu\neq\emptyset)]$ can be inferred from the proof proposed in \cite{nguyen2016cost} and $\Pr_{S\sim\vec{x}}[S\cap\nu\neq\emptyset]$ is the probabilty there is at least one node in $\nu$ activated as a seed, thus, we have $\Pr_{S\sim\vec{x}}[S\cap \nu\neq\emptyset]=1-\prod_{s\in\nu}(1-h_s(\vec{x}))=\mathcal{H}(\nu)$.
\end{proof}
\noindent
Let $N=\{\nu_1,\nu_2,\cdots,\nu_\theta\}$ be a collection of $\theta$ independent RN-sampling, by Equation (18), we have
\begin{equation}
\overline{\hat{f}_c}(\vec{x})=\frac{W}{\theta}\sum_{\nu\in N}\left(\mathcal{H}(\nu)\right)
\end{equation}
According to Theorem 10 and Theorem 11, $\underline{\hat{f}_c}(\vec{x})$ and $\overline{\hat{f}_c}(\vec{x})$ is an unbiased estimator of $\underline{f_c}(\vec{x})$ and $\overline{f_c}(\vec{x})$ for any fixed $\theta$ and they are monotone and DR-submodular.

\subsection{Modified IMM on Lattice}
The unbiased estimators of our objective functions have been obtained in last subsection, here, we extend the IMM algorithm \cite{tang2015influence}, the state-of-the-art method for the IM problem, to design the solutions of lower bound and upper bound of our CAM problem. The core idea of IMM on IM problem: produce enough random reverse reachable set (Random RR-Set), where the node is selected uniformly and randamly, and then find the maximum coverage under the cardinality constraint by use of greedy algorithm. The IMM process can be divided into two stages as follows:
\begin{enumerate}
	\item Sampling Random RR-Sets: This stage generates enough random RR-set iteratively and independently and put then into $\mathcal{R}$ until satisfying a certain stopping condition.
	\item Node selection: This stage adopts standard greedy method to drive a size-k seed set that covers sub-maximum number of RR-Sets in $\mathcal{R}$.
\end{enumerate}

Extended to our problem, we generate enough RE-sampling for lower bound or RN-sampling for upper bound fitst, then the lattice-greedy algorithm on these RE-sampling or RN-sampling is adopted to get the sub-optimal strategy marketing $\vec{x}$. Let us introduce the node selection first. Let $M=\{\mu_1,\mu_2,\cdots,\mu_\theta\}$ be a collection of $\theta$ independent RE-sampling and $N=\{\nu_1,\nu_2,\cdots,\nu_\theta\}$ be a collection of $\theta$ independent RN-sampling. The node selection is shown in Algorithm \ref{a2}, which is a $(1-1/e)$-approximate solution to the estimator of upper and lower bound.

In the first stage, we can use the sampling procedure similar to IMM, but need some modifications. For the lower bound, these modifications are: (1) we replace the number of node $n$ with $T$, where $T=\sum_{(u,v)\in E}A_{uv}$; (2) we use lattice-greedy algorithm, Algorithm \ref{a2}, on RE-sampling instead of greedy algorithm on RR-set; and (3) we replace $\log\binom{n}{k}$ with $\min(kt^{-1}\log d,d\log(kt^{-1}))$ in the two parameters $\lambda'$ and $\lambda^*$ \cite{chen2018scalable}. We have
\begin{flalign}
	&\alpha=\sqrt{\ell\log T+\log 2}\\
	&\beta=\sqrt{(1-1/e)(\min(kt^{-1}\log d,d\log(kt^{-1}))+\alpha^2)}
\end{flalign}

\begin{algorithm}[!t]
	\caption{\textbf{lattice-Greedy $(\underline{\hat{f}_c}(\overline{\hat{f}_c}),M(N),\mathcal{X},k)$}}\label{a2}
	\begin{algorithmic}[1]
		\STATE Initialize: $\vec{x}^\circ=0$ and $c=0$
		\WHILE {$c<k$}
		\STATE $i^\circ\leftarrow\arg\max_{i\in [d]}\left(\underline{\hat{f}_c}(\overline{\hat{f}_c})(\vec{x}^\circ+t\vec{e}_i)-\underline{\hat{f}_c}(\overline{\hat{f}_c})(\vec{x}^\circ)\right)$
		\STATE $\vec{x}^\circ\leftarrow\vec{x}^\circ+t\vec{e}_{i^\circ}$
		\STATE $c\leftarrow c+t$
		\ENDWHILE
		\RETURN $\vec{x}^\circ$
	\end{algorithmic}
\end{algorithm}

\begin{algorithm}[!t]
	\caption{\textbf{sampling-LB $(G,\underline{\hat{f}_c},\mathcal{X},k,\varepsilon,\ell)$}}\label{a3}
	\begin{algorithmic}[1]
		\STATE Initialize: $M=\emptyset$, $LB=0$, $\varepsilon'=\sqrt{2}\varepsilon$
		\STATE Initialize: $M'=\emptyset$
		\STATE Let $\lambda'=(2+\frac{2}{3}\varepsilon')(\min(Bt^{-1}\log d,d\log(Bt^{-1}))+\ell\log T+\log\log_2 T)\cdot T/\varepsilon'^2$
		\STATE Let $\lambda^*=2T\cdot((1-\frac{1}{e})\cdot\alpha+\beta)^2/\varepsilon^2$
		\FOR {$i=1$ to $\log_2 T-1$}
		\STATE Let $y_i=T/2^i$
		\STATE Let $\theta_i=\lambda'/y_i$, where $\lambda'$ is defined above
		\WHILE {$|M|\leq\theta_i$}
		\STATE $\mu\leftarrow$ RE-sampling $(G)$
		\STATE $M\leftarrow M\cup\{\mu\}$
		\ENDWHILE
		\STATE $\vec{x}^\circ\leftarrow$ lattice-Greedy $(\underline{\hat{f}_c},M,\mathcal{X},k)$
		\IF {$\underline{\hat{f}_c}(\vec{x}^\circ)\geq(1+\varepsilon')\cdot y_i$}
		\STATE $LB\leftarrow\underline{\hat{f}_c}(\vec{x}^\circ)/(1+\varepsilon')$
		\STATE \textbf{break}
		\ENDIF
		\ENDFOR
		\STATE $\theta\leftarrow\lambda^*/LB$
		\WHILE {$|M'|\leq\theta$}
		\STATE $\mu\leftarrow$ RE-sampling $(G)$
		\STATE $M'\leftarrow M'\cup\{\mu\}$
		\ENDWHILE
		\RETURN $M'$
	\end{algorithmic}
\end{algorithm}
\begin{algorithm}[!t]
	\caption{\textbf{IMM-LB $(G,\underline{\hat{f}_c},\mathcal{X},k,\varepsilon,\ell)$}}\label{a4}
	\begin{algorithmic}[1]
		\STATE $M'\leftarrow$ sampling-LB $(G,\underline{\hat{f}_c},\mathcal{X},k,\varepsilon,\ell)$
		\STATE $\vec{x}_L\leftarrow$ lattice-Greedy $(\underline{\hat{f}_c},M',\mathcal{X},B)$
		\RETURN $\vec{x}_L$
	\end{algorithmic}
\end{algorithm}

\noindent
Then, the sampling procedure for lower bound, sampling-LB, can be shown in Algorithm \ref{a3}, where $\varepsilon$ is accuracy and $\ell$ is confidence. Chen has told us that there is an issue \cite{chen2018issue} in original IMM algorihtm \cite{tang2015influence} and gave us two workarounds \cite{chen2018scalable}. We adopt the first workaround, line 19 to 22 in Algorithm \ref{a3}, that is more simple and straightforward. The IMM-LB algorithm is shown in Algorithm \ref{a4}.

\begin{thm}
	The solution $\vec{x}_L$ returned by Algorithm \ref{a4} is a $(1-1/e-\varepsilon)$-approximation of the upper bound of CAM problem with at least $1-1/T^\ell$ probability.
\end{thm}

To the original problem, we have known that $\hat{f}_c(\vec{x})$ is an unbiased estimator of $f_c(\vec{x})$. Based on the collection $M'$ generated in Algorithm \ref{a4}, we can use it to get solution $\vec{x}_A$ by calling lattice-Greedy $(\hat{f}_c,M',\mathcal{X},k)$, because they are all relying on RE-sampling. Here, $\vec{x}_A$ is a heuristic solution, no any theoretical guarantee, to the CAM problem.

\begin{algorithm}[!t]
	\caption{\textbf{sampling-UB $(G,\overline{\hat{f}_c},\mathcal{X},k,\varepsilon,\ell)$}}\label{a5}
	\begin{algorithmic}[1]
		\STATE Initialize: $N=\emptyset$, $LB=0$, $\varepsilon'=\sqrt{2}\varepsilon$
		\STATE Initialize: $N'=\emptyset$
		\STATE Let $\lambda'=(2+\frac{2}{3}\varepsilon')(\min(Bt^{-1}\log d,d\log(Bt^{-1})+\ell\log W+\log\log_2 W)\cdot W/\varepsilon'^2$
		\STATE Let $\lambda^*=2W\cdot((1-\frac{1}{e})\cdot\alpha+\beta)^2/\varepsilon^2$
		\FOR {$i=1$ to $\log_2 W-1$}
		\STATE Let $y_i=T/2^i$
		\STATE Let $\theta_i=\lambda'/y_i$, where $\lambda'$ is defined above
		\WHILE {$|N|\leq\theta_i$}
		\STATE $\nu\leftarrow$ RN-sampling $(G)$
		\STATE $N\leftarrow N\cup\{\nu\}$
		\ENDWHILE
		\STATE $\vec{x}^\circ\leftarrow$ lattice-Greedy $(\overline{\hat{f}_c},N,\mathcal{X},k)$
		\IF {$\overline{\hat{f}_c}(\vec{x}^\circ)\geq(1+\varepsilon')\cdot y_i$}
		\STATE $LB\leftarrow\overline{\hat{f}_c}(\vec{x}^\circ)/(1+\varepsilon')$
		\STATE \textbf{break}
		\ENDIF
		\ENDFOR
		\STATE $\theta\leftarrow\lambda^*/LB$
		\WHILE {$|N'|\leq\theta$}
		\STATE $\nu\leftarrow$ RN-sampling $(G)$
		\STATE $N'\leftarrow N'\cup\{\nu\}$
		\ENDWHILE
		\RETURN $N'$
	\end{algorithmic}
\end{algorithm}

\begin{algorithm}[!t]
	\caption{\textbf{IMM-UB $(G,\overline{\hat{f}_c},\mathcal{X},k,\varepsilon,\ell)$}}\label{a6}
	\begin{algorithmic}[1]
		\STATE $N'\leftarrow$ sampling-UB $(G,\overline{\hat{f}_c},\mathcal{X},k,\varepsilon,\ell)$
		\STATE $\vec{x}_U\leftarrow$ lattice-Greedy $(\overline{\hat{f}_c},N',\mathcal{X},k)$
		\RETURN $\vec{x}_U$
	\end{algorithmic}
\end{algorithm}

For the upper bound, the modifications are similar to that of lower bound, but (1) we replace the number of node $n$ with $W$, where $W=\sum_{u\in V}w(u)$; and (2) we use lattice-greedy algorithm, Algorithm \ref{a2}, on RN-sampling. That is,
\begin{flalign}
	&\alpha'=\sqrt{\ell\log W+\log 2}\\
	&\beta'=\sqrt{(1-1/e)(\min(kt^{-1}\log d,d\log(kt^{-1}))+\alpha'^2)}
\end{flalign}
Then, the sampling procedure for upper bound, sampling-UB, can be shown in Algorithm \ref{a5}, where $\varepsilon$ is accuracy and $\ell$ is confidence. The IMM-UB algorithm is shown in Algorithm \ref{a6} similarly.

\begin{thm}
	The solution $\vec{x}_U$ returned by Algorithm \ref{a6} is a $(1-1/e-\varepsilon)$-approximation of the upper bound of CAM problem with at least $1-1/W^\ell$ probability.
\end{thm}

\subsection{Sandwich Approximation Framework}
To optimize non-submodular function, there is no universal technique to approximate it within constant approximation ratio. Lu et al. \cite{lu2015competition} provided a sandwich approximation framework to us, where a data-dependent approximation ratio can be obtained by approximating the upper bound and lower bound that are monotone and submodular. It can be extended to solve our monotone but not DR-submodular objective function. First, we get a $(1-1/e-\varepsilon)$-approximate solution to the lower bound by calling IMM-LB, during that, we record the immediate collection of RE-sampling $M'$. Then, we use this $M'$ as the input of lattice-greedy to find a heuristic solution to the original problem. Finally, we get a  $(1-1/e-\varepsilon)$-approximate solution to the upper bound by calling IMM-UB and return the best one to the original problem. It is shown in Algorithm \ref{a7}.

\begin{thm}
	Let $\vec{x}_{sand}$ be the marketing strategy returned by Algorithm \ref{a7}, then we have $f_c(\vec{x}_{sand})\geq$
	\begin{equation}
	\max\left\{\frac{f_c(\vec{x}_U)}{\overline{f_c}(\vec{x}_U)},\frac{\underline{f_c}(\vec{x}_L^*)}{f_c(\vec{x}_A^*)}\right\}\frac{1-\gamma}{1+\gamma}\left(1-\frac{1}{e}-\varepsilon\right)f_c(\vec{x}_A^*)
	\end{equation}
	where $\vec{x}_L^*$ is the optimal solution to maximize the lower bound and $\vec{x}_A^*$ is the optimal solution of the CAM problem.
\end{thm}
\begin{proof}
	Let $\vec{x}_U^*$ be the optimal solution to maximize the upper bound. For the upper bound, we have
	\begin{flalign}
		f_c(\vec{x}_U)&=\frac{f_c(\vec{x}_U)}{\overline{f_c}(\vec{x}_U)} \overline{f_c}(\vec{x}_U)\geq\frac{f_c(\vec{x}_U)}{\overline{f_c}(\vec{x}_U)}\left(1-\frac{1}{e}-\varepsilon\right)\overline{f_c}(\vec{x}_U^*)\nonumber\\
		&\geq\frac{f_c(\vec{x}_U)}{\overline{f_c}(\vec{x}_U)}\left(1-\frac{1}{e}-\varepsilon\right)\overline{f_c}(\vec{x}_A^*)\nonumber\\
		&\geq\frac{f_c(\vec{x}_U)}{\overline{f_c}(\vec{x}_U)}\left(1-\frac{1}{e}-\varepsilon\right)f_c(\vec{x}_A^*)\nonumber
	\end{flalign}
	For the lower bound, we have
	\begin{flalign}
		f_c(\vec{x}_L)&\geq\underline{f_c}(\vec{x}_L)\geq\left(1-\frac{1}{e}-\varepsilon\right)\underline{f_c}(\vec{x}_L^*)\nonumber\\
		&\geq\frac{\underline{f_c}(\vec{x}_L^*)}{f_c(\vec{x}_A^*)}\left(1-\frac{1}{e}-\varepsilon\right)f_c(\vec{x}_A^*)\nonumber
	\end{flalign}
	Let $\vec{x}_{max}=\arg\max_{\vec{x}\in\{\vec{x}_L,\vec{x}_A,\vec{x}_U\}}f_c(\vec{x})$, that is,
	\begin{equation*}
		f_c(\vec{x}_{max})\geq\max\left\{\frac{f_c(\vec{x}_U)}{\overline{f_c}(\vec{x}_U)},\frac{\underline{f_c}(\vec{x}_L^*)}{f_c(\vec{x}_A^*)}\right\}\left(1-\frac{1}{e}-\varepsilon\right)f_c(\vec{x}_A^*)
	\end{equation*}
	According to Theorem 2, $\dot{f}_c(\vec{x})$ is a $(\gamma,\delta)$-Estimation of $f_c(\vec{x})$ given a marketing strategy $\vec{x}$. Then, $\vec{x}_{sand}=\arg\max_{\vec{x}\in\{\vec{x}_L,\vec{x}_A,\vec{x}_U\}}\dot{f}_c(\vec{x})$, if $\vec{x}_{sand}\neq\vec{x}_{max}$, we have $(1+\gamma)f_c(\vec{x}_{sand})\geq(1-\gamma)f_c(\vec{x}_{max})$. Thus, the Inequality (24) is established.
\end{proof}

\begin{algorithm}[!t]
	\caption{\textbf{Sandwich Approximation Framework}}\label{a7}
	\begin{algorithmic}[1]
		\STATE $\vec{x}_L\leftarrow$ IMM-LB $(G,\underline{\hat{f}_c},\mathcal{X},k,\varepsilon,\ell)$ // Record the $M'$ returned by sampling-LB here
		\STATE $\vec{x}_A\leftarrow$ lattice-Greedy $(\hat{f}_c,M',\mathcal{X},k)$
		\STATE $\vec{x}_U\leftarrow$ IMM-UB $(G,\overline{\hat{f}_c},\mathcal{X},k,\varepsilon,\ell)$
		\STATE $\vec{x}_{sand}\leftarrow\arg\max_{\vec{x}\in\{\vec{x}_L,\vec{x}_A,\vec{x}_U\}}\dot{f_c}(\vec{x})$, where $\dot{f_c}(\vec{x})$ can be computed by $\dot{f}_d(\widetilde{V}-V|\widetilde{G})$ on constructed graph $\widetilde{G}$ equivalently, shown as Remark 2.
		\RETURN $\vec{x}_{sand}$
	\end{algorithmic}
\end{algorithm}

\section{Experiment}
In this section, we carry out several experiments on different datasets to validate the correctness and efficiency of our proposed algorithms. There are three datasets \cite{nr} used in our experiments: (1) Dataset-1: a co-authorship network, co-authorship among scientists to publish papers about network science; (2) Dataset-2: a Wiki network, who-voteson-whom network which come from the collection Wikipedia voting; (3) Dataset-3: A collaboration netwook extracted from Arxiv General Relativity. The statistics information of the three datasets is represented in table \ref{table_example}.

\begin{table}[h]
	\renewcommand{\arraystretch}{1.3}
	\caption{The statistics of three datasets}
	\label{table_example}
	\centering
	\begin{tabular}{|c|c|c|c|c|}
		\hline
		\bfseries Dataset & \bfseries n & \bfseries m & \bfseries Type & \bfseries Average degree\\
		\hline
		dataset-1 & 0.4K & 1.01K & directed & 4\\
		\hline
		dataset-2 & 1.0K & 3.15K & directed & 6\\
		\hline
		dataset-3 & 5.2K & 14.5K & directed & 5\\
		\hline
	\end{tabular}
\end{table}

\begin{figure}[!t]
	\centering
	\includegraphics[width=3.5in]{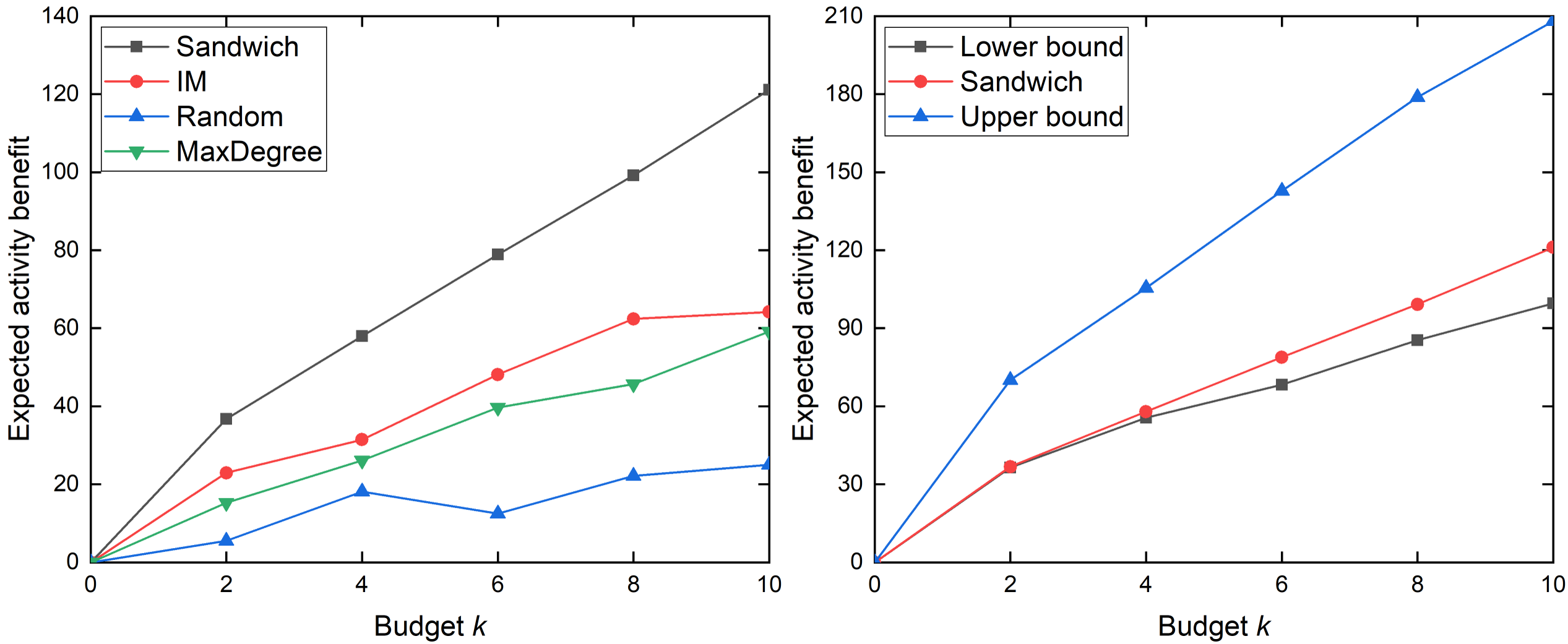}
	\\(a) Dataset-1
	\\${}$
	\includegraphics[width=3.5in]{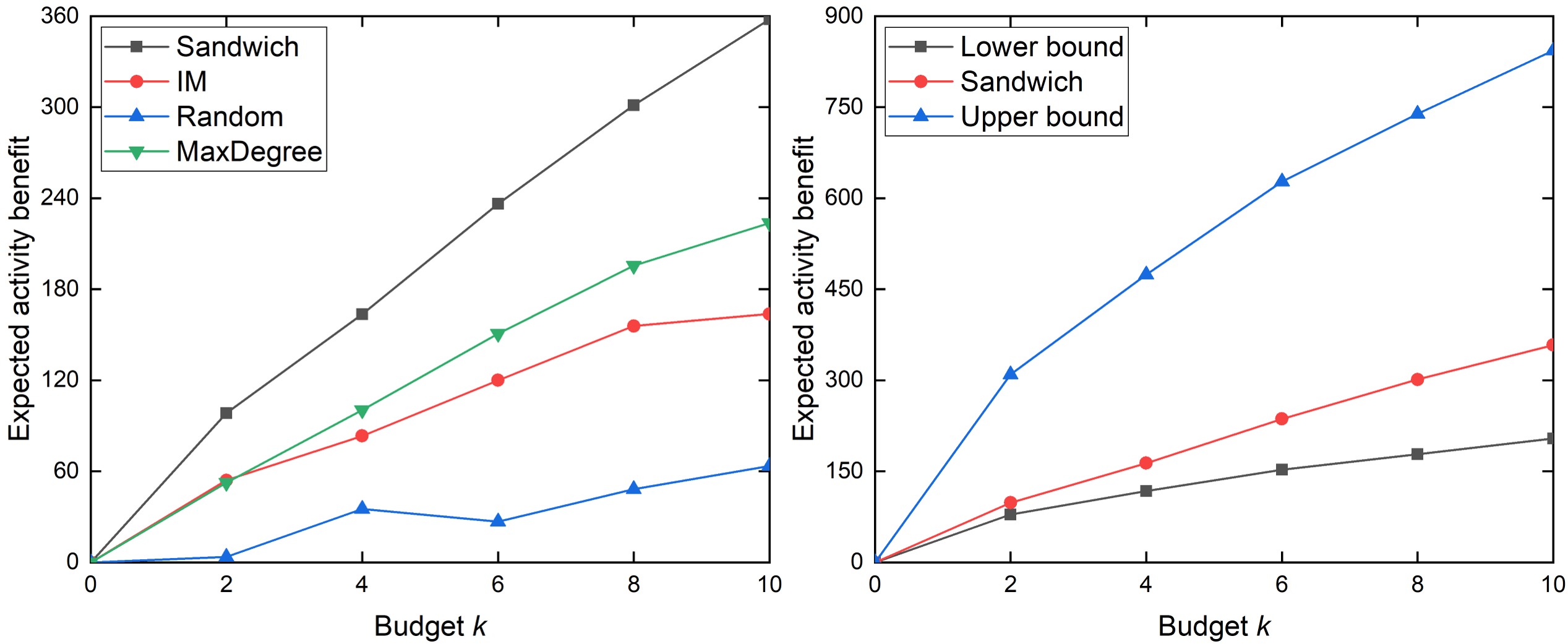}
	\\(b) Dataset-2
	\\${}$
	\includegraphics[width=3.5in]{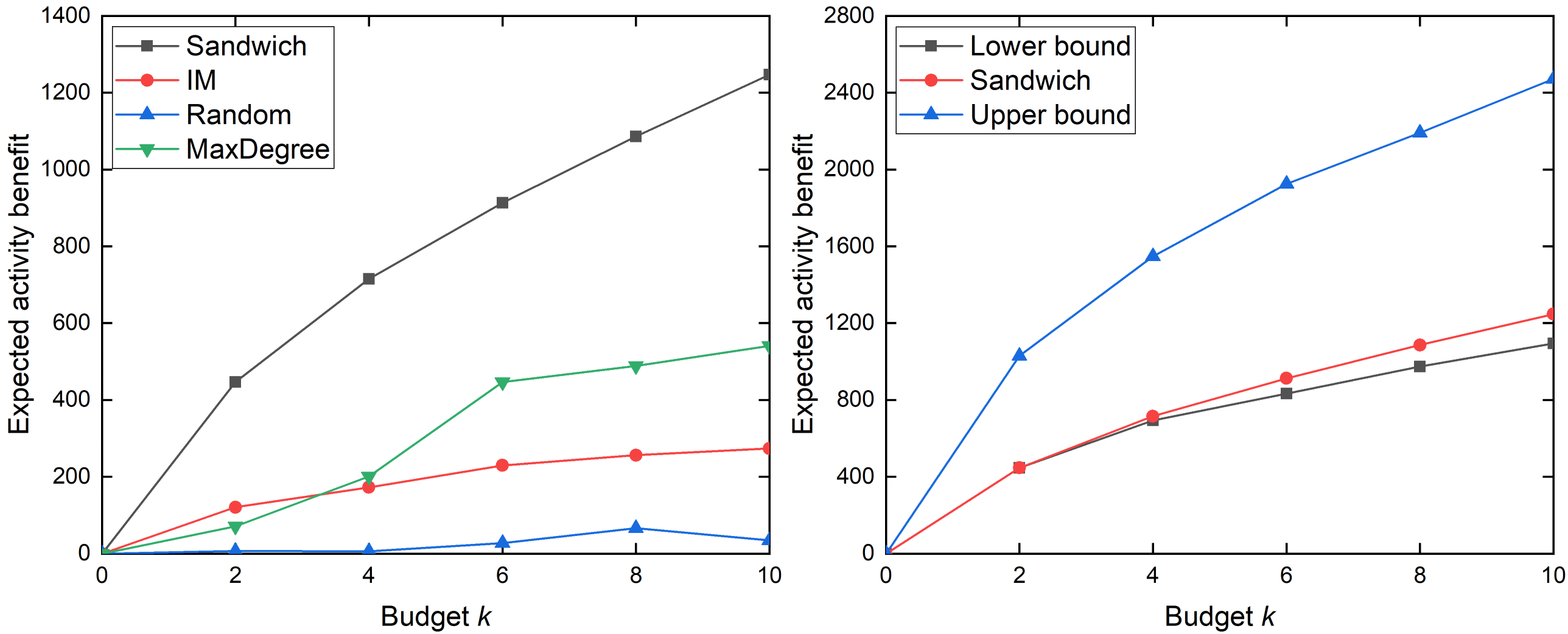}
	\\(c) Dataset-3
	\caption{Under the IC-model: left column is the performance comparison of different algorithms changes over budget $k$; right column is the result of sandwich approximation framework.}
	\label{fig1}
\end{figure}

\subsection{Experimental Settings}
The diffusion model of our proposed experiments relies on IC-model and LT-model. Under the IC-model, for each edge $(u,v)\in E$, the diffusion probability is set as $p_{uv}=1/|N^-(v)|$. Under the the LT-model, for each edge $e=(u,v)$, the weight is set as $b_{uv}=1/|N^-(v)|$. This setting is widely used by prior works about influence maximization. Given a marketing strategy $\vec{x}$, for each node $u\in V$, we have a strategy function $h_u(\vec{x})$. Here, we consider the case: independent strategy activation \cite{chen2018scalable}, where each component $x_j\in\vec{x}$ attempts to activate $u$ as seed independently. Then, we have
\begin{equation}
h_u(\vec{x})=1-\prod_{j\in[d]}(1-q_{uj}(x_j))
\end{equation}
where strategy $j\in[d]$ activate $u$ as seed with probability $q_{uj}(x_j)$. Chen et al. \cite{chen2018scalable} pointed out $h_u(\vec{x})$ is monotone and DR-submodular if $q_{uj}(x_j)$ is monotone and concave for each $j\in[d]$ and each node $u\in V$. In this experiment, we test personalized marketing scenario \cite{yang2016continuous}, where strategy function is defined as $h_u(\vec{x})=2x_u-x_u^2$ and $\vec{x}=(x_1,x_2,\cdots,x_n)$. It means that the probability that activates node $u$ as seed only depends on component $x_u$.

For our sandwich approximation framework, we set parameters of accuracy $\varepsilon=0.1$, confidence $\ell=1$ and granularity $t=0.2$. Besides, we set activity strength $A_{uv}=1$ for each edge $(u,v)\in E$ and Monte Carlo simulation $r=2000$. Then, we compare it with some commonly used baseline algorithms, which is summarized as follows: (1) IM: It returns the active nodes by lattice greedy algorithm to maximizing the influence spread, and then computes the activity benefit. (2) MaxDegree: It selects the node with the highest outdegree under the budget $k$. (3) Random: It selects a node $u$ randomly and increases its $x_u$ by $t$ until using up the budget $k$. 

\begin{figure}[!t]
	\centering
	\includegraphics[width=3.5in]{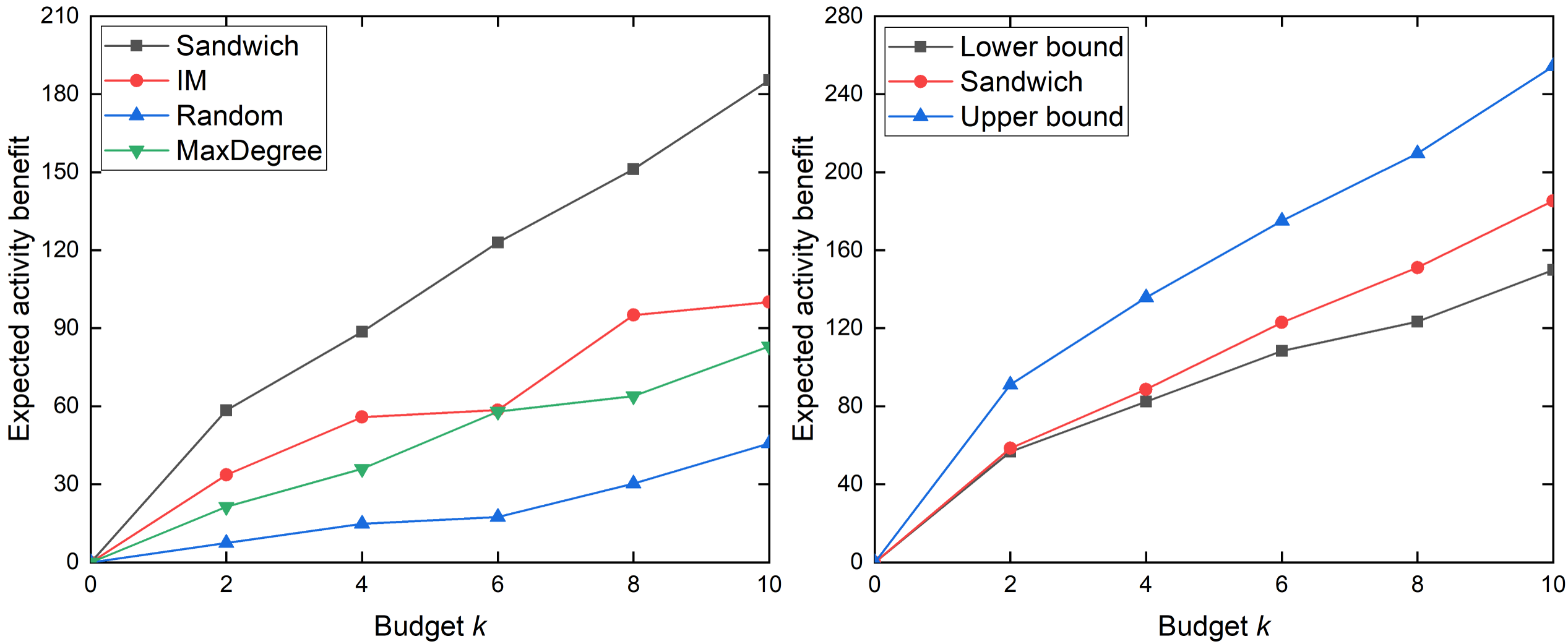}
	\\(a) Dataset-1
	\\${}$
	\includegraphics[width=3.5in]{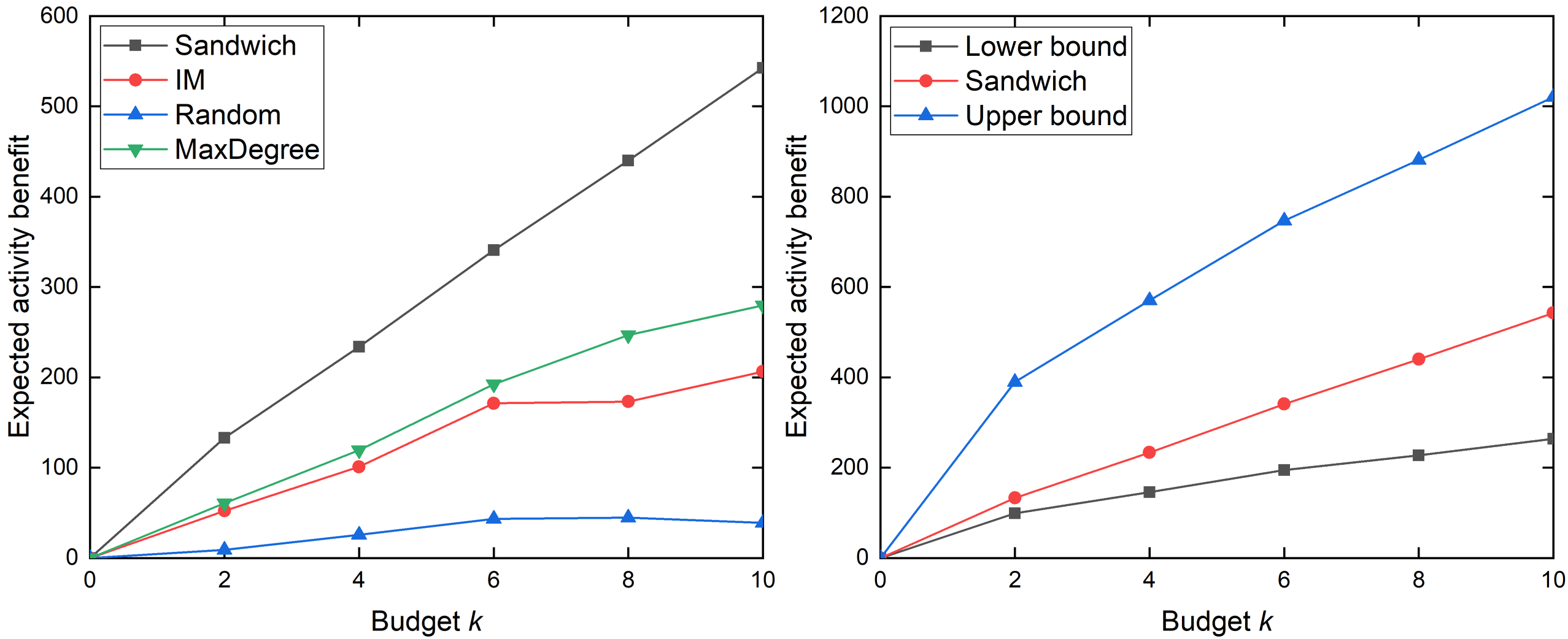}
	\\(b) Dataset-2
	\\${}$
	\includegraphics[width=3.5in]{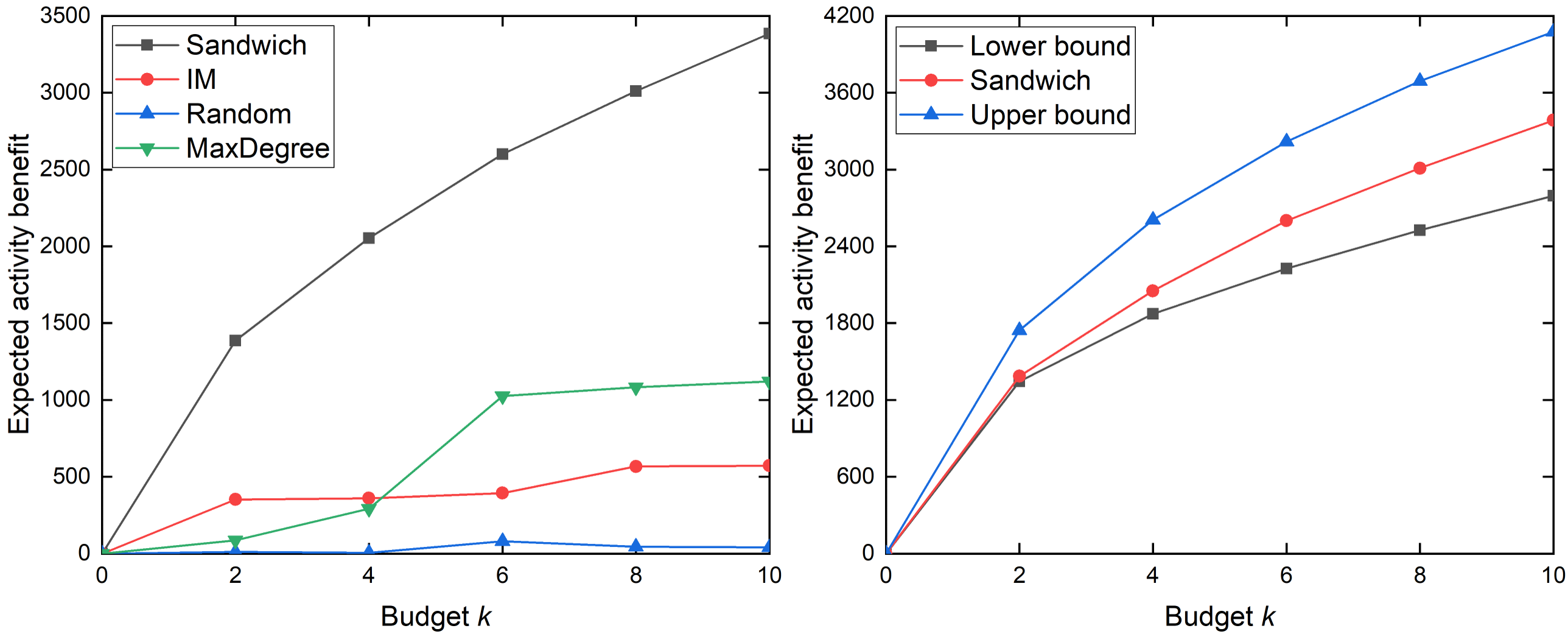}
	\\(c) Dataset-3
	\caption{Under the LT-model: left column is the performance comparison of different algorithms changes over budget $k$; right column is the result of sandwich approximation framework.}
	\label{fig2}
\end{figure}

\subsection{Experimental Results}
Fig. \ref{fig1} and Fig. \ref{fig2} draw the performance achieved by our sandwich method under the IC-model and LT-model, Algorithm \ref{a7}, and other heuristic algorithms. Theoretically, our sandwich method can guarantee an approximate bound, but others can not. From the left column of Fig. \ref{fig1} and Fig. \ref{fig2}, the total activity benefit returned by our sandwich method is always the best among all results returned by other algorithms. With the increasing size of dataset, the advantage of sandwich is more apparent. For IM and MaxDegree, which one is better? The answer is uncertain. For the dataset-1, IM is better than MaxDegree under the IC-model and LT-model. But for the dataset-2 and dataset-3, MaxDegree is better than IM. From the right column of Fig. \ref{fig1} and Fig. \ref{fig2}, it is observed that the expected activity benefit returned by sandwich approximation framework lies in between its upper bound and lower bound. Unitil now, the correctness and effectiveness of our algorithms have been tested and validated.

\section{Conclusion}
In this paper, we proposed the CAM problem to adapt to real scenario, general marketing strategy. It can be considered as maximization problem on lattice. We proved the hardness and gave a computing method for the objective function of CAM. This objective function is monotone but not DR-submodular and not DR-supermodular. We designed the unbiased sampling for it, its upper bound and lower bound. Adapted from IMM algorithm and sandwich approximation framework, a data-dependent approximation ratio can be obtained. The performance of the proposed algorithms is verified by experiments. The analysis of CAM problem is applicable to others which is a branch of maximization problem on lattice.

\section*{Acknowledgment}

This work is partly supported by National Science Foundation under grant 1747818.

\ifCLASSOPTIONcaptionsoff
  \newpage
\fi



%\bibliographystyle{IEEEtran}

\bibliographystyle{IEEEtran}
\bibliography{references}



%\vfill


\end{document}